\documentclass[a4paper,12pt]{elsarticle}
\usepackage{amssymb}
\usepackage{amsmath}
\usepackage{latexsym}
\usepackage{comment}
\usepackage{graphicx}
\usepackage{verbatim}
\usepackage[color=pink]{todonotes}
\graphicspath{{graphics/}}
\usepackage{tikz}
\usepackage[utf8]{inputenc}
\usepackage[T1]{fontenc}
\usepackage[spanish]{babel}

\newcommand*\circled[1]{\tikz[baseline=(char.base)]{
            \node[shape=circle,draw,inner sep=2pt] (char) {#1};}}
 \newtheorem{teor}{Theorem}
 \newtheorem{lema}{Lemma}

 \newtheorem{remark}{Remark}

\newcommand{\R}{\mathbb{R}}

\begin{document}
\begin{frontmatter}


\title{Multiform Longest Edge Bisection of Tetrahedra via Sextuple Permutations}

\author[agus]{Agusti\'n Trujillo}
\ead{agustin.trujillo@ulpgc.es}
\address[agus]{University of Las  Palmas de Gran Canaria (ULPGC), Spain }

\author[jose]{Jos\'e Pablo Su\'arez\corref{mycorrespondingauthor}}
\ead{jose.suarez@ulpgc.es}
\cortext[mycorrespondingauthor]{Corresponding author}
\address[jose]{IUMA. University of Las  Palmas de Gran Canaria (ULPGC), Spain }

\author[tania]{Tania Moreno-García\corref{l3}}
\ead{taniamgholguin@gmail.com}
\address[tania]{Facultad de Informatica y Matematica, Universidad de Holguin, Holguin, Cuba}

\begin{abstract}
We introduce a new formulation of the Longest Edge Bisection (LEB) of tetrahedra entirely in sextuple space $\mathbb{R}^6$, where tetrahedra are represented by the squares of their edge lengths. This representation renders the LEB refinement equations fully linear and eliminates the need for coordinate-based data structures.

A central difficulty in three-dimensional LEB arises when a tetrahedron possesses multiple longest edges, making the refinement rule intrinsically multivalued. We formalize this phenomenon through the notion of \emph{Multiform Longest Edge Bisection} (MLEB), which systematically explores all admissible longest-edge choices. To encode this multivalued behavior, we introduce the concept of \emph{bisection patterns}, defined as sequences of sextuple permutations governing the refinement process.

We prove that the set of sextuples sharing a common LEB pattern forms a convex region in $\mathbb{R}^6$. For structurally significant families of tetrahedra, including the $R_1^+$ family and the Liu–Joe family, we show that the infinite refinement tree collapses into a finite directed graph with eight states. Remarkably, both families are governed by the same graph, differing only in their initial state.

This directed-graph formulation provides a unified combinatorial description of the refinement process and offers an efficient computational framework for deep iterative LEB analysis.
\end{abstract}


\end{frontmatter}

\section{Introduction}

The Longest Edge Bisection (LEB) is one of the most widely used refinement strategies for simplicial meshes in two and three dimensions. In three dimensions, the subdivision of a tetrahedron is performed by connecting the midpoint of its longest edge to the opposite edge, generating two child tetrahedra. Iterating this process produces a refinement tree whose geometric and combinatorial properties have been extensively studied in the context of adaptive finite element methods.

Traditionally, LEB algorithms rely on coordinate-based representations of tetrahedra. In contrast, we adopt a purely edge-based description, representing each tetrahedron by a sextuple in $\mathbb{R}^6$ consisting of the squares of its six edge lengths. This representation is invariant under rigid motions and eliminates orientation dependencies. Moreover, when squared edge lengths are used, the refinement equations become linear, significantly simplifying both theoretical analysis and computational implementation. 
The idea of representing a tetrahedron by its edges is not new, see for example \cite{WirDre} and the works referenced therein. Lastly new focus on the edge sextuples for the LEB for tetrahedra has been presented in \cite{SuTrTa, AMC_2024}. 

Some interesting properties of sextuples are as follows. As a tetrahedron has exactly six edges, it can be fully represented by a sextuple of edge lengths (six positive real numbers). In contrast, four vertex coordinates are needed, which requires 12 real numbers. Edge-length representation is independent of orientation, position, or coordinate system. Then sextuples are useful for comparing tetrahedra shapes purely based on geometry, \cite{SuTrTa, AMC_2024, PadronTrujilloSuarez2025}. 

Most common geometric properties such as volume, face areas, and dihedral angles can be computed directly from edge lengths using formulas like the Cayley-Menger determinant.

A subtle but fundamental difficulty in the three-dimensional LEB arises when a tetrahedron possesses two or more edges of equal maximal length. In such cases, the refinement edge is not uniquely determined, and the bisection rule becomes intrinsically multivalued. Classical approaches, including those in \cite{Sik, Kos, RivLev, BeGiRo, HaKoKri_2014}, implicitly select one admissible longest edge according to a normalization convention, but they do not explicitly analyze the combinatorial branching induced by ties. However, a complete understanding of similarity classes, degeneration phenomena, and long-term refinement behavior requires exploring all admissible choices.

In \cite{Casado2015} a  global optimization  method is employed to overcome these difficulties in $n$-simplex, although it shows high computational cost of solving this combinatorial problem.

In this work we formalize this situation by introducing the notion of \emph{Multiform Longest Edge Bisection} (MLEB). Instead of enforcing a single deterministic choice, we interpret the LEB as a multivalued refinement process whenever equal longest edges occur. To encode this process algebraically, we introduce \emph{bisection patterns}, defined as sequences of sextuple permutations that determine which edge is bisected at each step. This leads naturally to a directed-graph representation of the refinement dynamics.

Our main contributions can be summarized as follows. We formulate the Longest Edge Bisection entirely in $\mathbb{R}^6$ using sextuples of squared edge lengths, which leads to a fully linear refinement scheme and removes the need for coordinate-based representations. We introduce the concept of Multiform Longest Edge Bisection (MLEB), providing a rigorous framework to handle the intrinsic multivalued nature of the refinement when equal longest edges occur. Within this framework, we define bisection patterns as permutation-driven rules governing the refinement process and show that the set of sextuples sharing a common pattern forms a convex subset of $\mathbb{R}^6$. Furthermore, we prove that for important families of tetrahedra, including the $R_1^+$ family and the Liu–Joe family, the infinite refinement tree collapses into a finite directed graph with eight states, and that both families are governed by exactly the same graph structure, differing only in the initial state. This directed-graph formulation provides a compact combinatorial description of the refinement process and yields an efficient computational framework for deep iterative LEB analysis.

\section{Representation of tetrahedra in $\R^6$}

We first see the case in $\R^3$ for triangles. Let $x,y,z$ be the squares of the lengths of the sides of an arbitrary triangle. Let $\textbf{P}=(x,y,z)$ be a point in $\R^3$ that represents this triangle. According to the \emph{triangle inequality}, we know that the values $x,y,z$ must satisfy the following conditions:
\begin{equation*}
\label{eq_triangle_ineq}
\sqrt{x}+\sqrt{y}\geq\sqrt{z}
\end{equation*}
\begin{equation}
\sqrt{x}+\sqrt{z}\geq\sqrt{y}
\end{equation}
\begin{equation*}
\sqrt{y}+\sqrt{z}\geq\sqrt{x}
\end{equation*}
Although these should be considered as strict inequalities, we have added equality to them in order to also consider triangles with zero area. We will refer to these triangles as \emph{degenerated triplets}.

Therefore, not all points in $\R^3$ represent a real triangle. In Figure \ref{trianglesInR3} the boundary of the region in space where all the points representing valid triangles are located is shown. We will call these points \emph{triangular triplets}. The points outside this region will be referred to as \emph{non-triangular triplets}. The boundary between the two sets is formed by the \emph{degenerated triplets}, which represent triangles with zero area. This statement is proven at the end of this section in Theorem \ref{teorema1} for the case of representing tetrahedra as sextuples in $\mathbb{R}^6$. 


\begin{figure}[h]
\centering
\leavevmode 
\includegraphics[scale=.4]{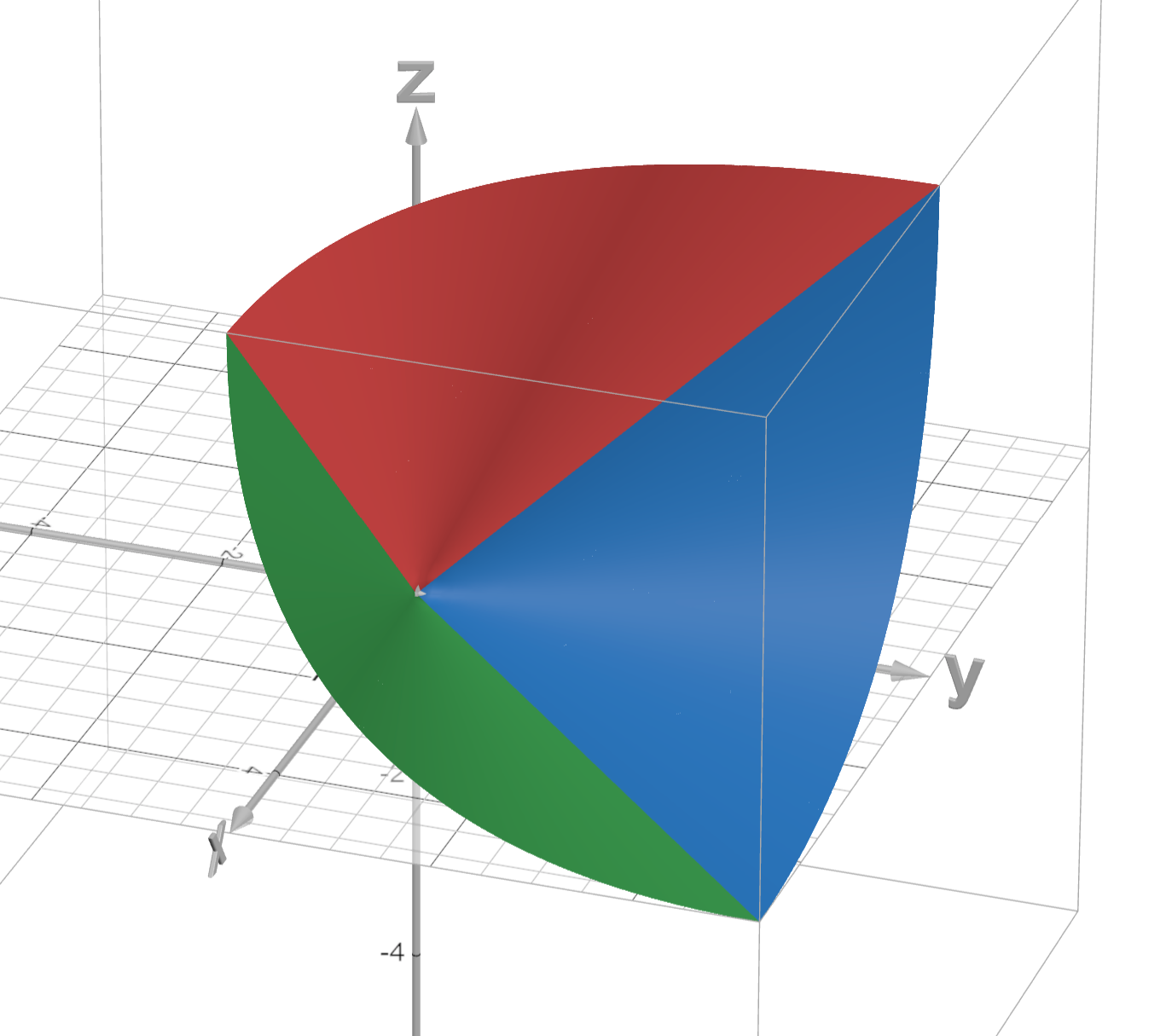}
\caption{Border of the region of triangular triplets. Points on this surface represent degenerated triangles.}
\label{trianglesInR3}
\end{figure}


Following the same nomenclature, let us represent an arbitrary tetrahedron by the squares of the lengths of its 6 edges. Let $\R^6$ be the real linear space with the base $\{{\textbf{e}_1},\,\textbf{e}_2,\,\textbf{e}_3,\,\textbf{e}_4,\,
\textbf{e}_5,\,\textbf{e}_6\}$, where $\textbf{e}_1=(1,\,0,\,0,\,0,\,0,\,0)$, $\textbf{e}_2=(0,\,1,\,0,\,0,\,0,\,0)$,..., 
$\textbf{e}_6=(0,\,0,\,0,\,0,\,0,\,1)$.\\
Then any element $\textbf{x}=(x_1,\,x_2,\,x_3,\,x_4,\,x_5,\,x_6)$ can be written as 
\[
\textbf{x}=x_1\textbf{e}_1+x_2\textbf{e}_2+x_3\textbf{e}_3+x_4\textbf{e}_4+x_5\textbf{e}_5+x_6\textbf{e}_6.
\]
We will say that an element $\textbf{x}\in\R^6$ is a \emph{tetrahedral sextuple} if there 
exists a tetrahedron, $\textbf{T}(\textbf{x})$,  such that the components of $\textbf{x}$ are the squares of its edge lengths, and furthermore the pairs $(x_1,x_6)$, $(x_2,x_4)$ and $(x_3,x_5)$ are related to pairs of opposite edges, with $(x_1,x_2,x_3)$ being one of the faces of the tetrahedron. This specific order of the components of the sextuple has been extracted from \cite{SuTrTa}.

Two conditions must be met for a sextuple $\textbf{x}$ to be a \emph{tetrahedral sextuple}, following the theorem in \cite{WirDre}. The first condition is that $\textbf{x}$ must be \emph{facial}, that is, the four faces of the tetrahedron $\textbf{T}(\textbf{x})$ must satisfy the \emph{triangle inequality} (equation \ref{eq_triangle_ineq}). If the three inequalities for the four faces of $\textbf{T}(\textbf{x})$ hold strictly, that is, none of them are equal, we will say that $\textbf{x}$ is \emph{strictly facial}.

The second condition is that the \emph{Cayley-Menger determinant} must be greater than or equal to zero. The function that assign the tetrahedron volume through the Cayley-Menger determinant can be extended naturally to any sextuple in $R^6$. This is because the determinant can be formulated to any sextuple, even in the case that the determinant cannot be interpreted as a geometric volume. This function, which we shall call the \emph{generalized volume function} $V(\textbf{x})$, is defined as follows:

\begin{equation}
V(\textbf{x})=V(x_1,\,x_2,\,\cdots,\, x_6)=\frac{1}{288}\begin{vmatrix}
0 & 1 & 1 & 1 & 1 \\
1 & 0 & x_1 & x_2 & x_5 \\
1 & x_1 & 0 & x_3 & x_4\\
1 & x_2 & x_3 & 0 & x_6\\
1 & x_5 & x_4 & x_6 & 0
\end{vmatrix}
\label{Determinant}
\end{equation}



The table in Figure \ref{fig:table_sextuples-set} shows the different situations of an element $\textbf{x}\in\R^6$ based on the value of the two indicated conditions. If $\textbf{x}$ is \emph{strictly facial} and its determinant is positive, we will say that $\textbf{x}$ is a \emph{shapely sextuple}, representing a tetrahedron with positive volume (in green color). If the determinant is zero, and $\textbf{x}$ is facial or strictly facial, then $\textbf{x}$ is a \emph{degenerated sextuple} (in yellow color). Finally, if $\textbf{x}$ is non-facial or its determinant is negative, then it is a \emph{non-tetrahedral sextuple} (in red color).

\begin{figure}[h]
\centering
\leavevmode 
\includegraphics[scale=.55]{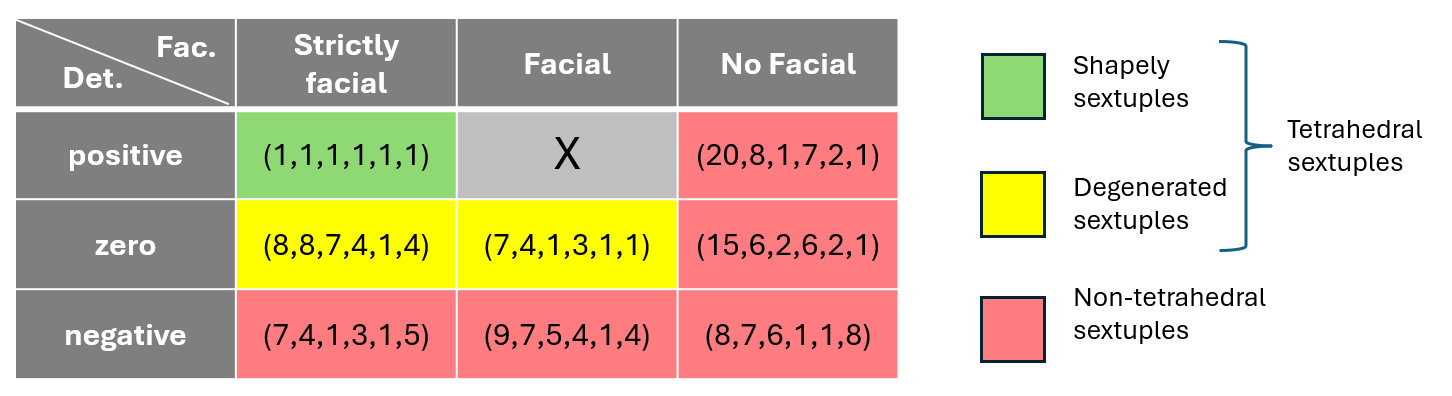}
\caption{Characterization of sextuples in $\R^6$ based on whether they are facial and the sign of the Cayley-Menger determinant. Each cell includes a sextuple as an example. The case of facial sextuples with a positive determinant is empty because it is not possible.}
\label{fig:table_sextuples-set}
\end{figure}

Also in this table, an example of a sextuple for each cell is shown. For instance, the sextuple $(1,1,1,1,1,1)$ is a \emph{shapely sextuple}, as it represents a regular tetrahedron with side length $1$, which is strictly facial and has positive volume. The sextuple $(8,8,7,4,1,4)$ is a \emph{degenerated sextuple}, because it is strictly facial and its determinant is zero. This sextuple has its four vertices coplanar, and therefore its volume is zero (see Figure \ref{fig:degenerated_sextuples}a). The sextuple $(7,4,1,3,1,1)$ is also a \emph{degenerated sextuple}, because it is facial but not strictly, and its determinant is zero. This sextuple has one face with zero area, meaning three collinear vertices, and therefore its volume is also zero (see Figure \ref{fig:degenerated_sextuples}b). Finally, there is no sextuple that is facial but not strictly, and at the same time has a positive determinant. This is not possible because having 3 collinear vertices makes it impossible to have a positive volume.

\begin{figure}[h]
\centering
\leavevmode 
\includegraphics[scale=.35]{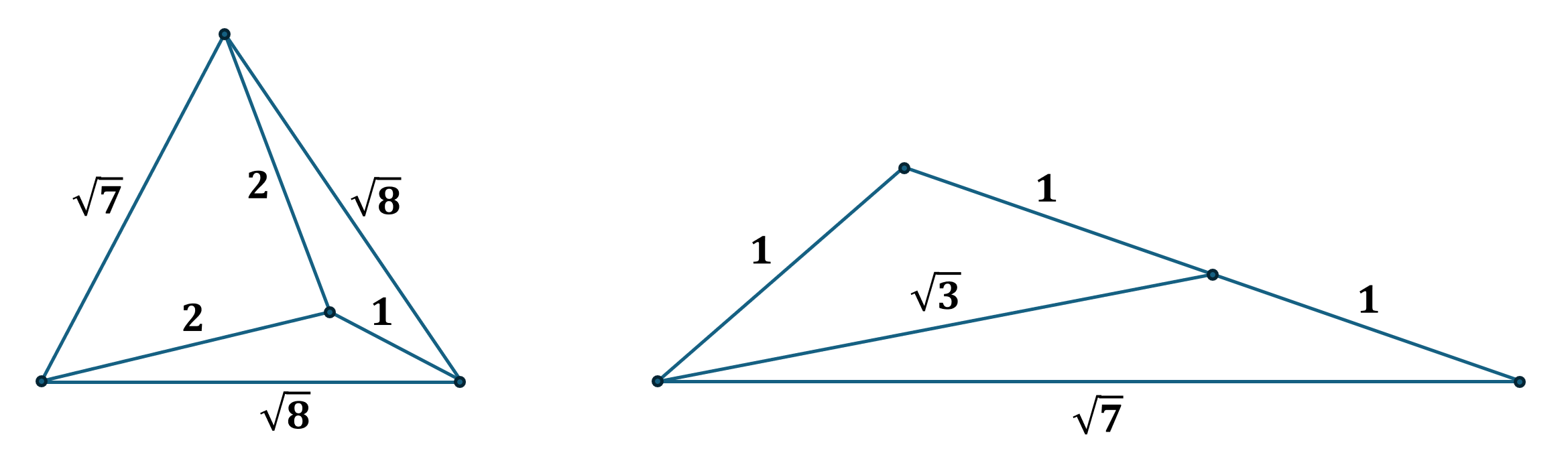}
\caption{Two examples of degenerated sextuples: $(8,8,7,4,1,4)$ on the left side; $(7,4,1,3,1,1)$ on the right side}
\label{fig:degenerated_sextuples}
\end{figure}

Therefore, we can draw an analogy with the case of representing triangles as points in $\R^3$, and state that not all points in $\R^6$ represent a real tetrahedron. Specifically, all sextuples that are not facial or whose Cayley-Menger determinant is negative, represent the region of \emph{non-tetrahedral sextuples} (red color in Figure \ref{fig:table_sextuples-set}). All remaining sextuples form the region of \emph{tetrahedral sextuples}. The boundary of both regions in $\R^6$ is a hypersurface
formed by the \emph{degenerated sextuples}, which represent tetrahedra with zero volume (yellow color in the table). The interior of the set of \emph{tetrahedral sextuples} is composed of the \emph{shapely sextuples}, which have positive volume (green color in the table).

\begin{lema}\label{closet} The set of all facial sextuples in $\R^6$ is a closed set.
\end{lema}

\noindent \textbf{Proof:}
  Let be  $X=(X_1,\,X_2,\cdots,\, X_6)$ a sextuple in $\R^6$ with $X_i \geq 0$ $\forall i \in {1,2,\cdots,\,6 } $. We can write the condition of being facial as following:

  \begin{eqnarray}\label{eq3}
  \begin{split}
      \sqrt X_1 + \sqrt X_2 & \geq \sqrt X_3 \\
    \sqrt X_2 + \sqrt X_3 & \geq \sqrt X_1   \\
     \sqrt X_1 + \sqrt X_3 & \geq \sqrt X_2  \\
     \cdots \\
     \sqrt X_5 + \sqrt X_6 & \geq \sqrt X_2 
     \end{split}
  \end{eqnarray}

Let $g_{ijk}(X)= \sqrt X_i + \sqrt X_j - \sqrt X_k  $ where $(i,j,k) \in (1,\,2,\cdots,\, 6)$. Note that $g_{ijk}(X)$ are continuous functions.
Given that the pre-image of a closed set under a continuous function is closed, it follows that the pre-image of $[0, \infty)$ under $g_{ijk}(X)$ is closed.
The system in Equation \ref{eq3} is made up finite inequalities in the form $g_{ijk}(X) \geq 0$.
Since the intersection of a finite number of closed sets is closed, it follows that the set of facial sextuples with non-negative components is closed. Sextuples with at least one negative component are not facial (by definition). Then, the theorem is proved.
\hfill \qed

\begin{lema}\label{DS2} The set of the sextuples $S$ such that $D(S) \geq 0$ is closed.
\end{lema}
\noindent \textbf{Proof:} 

We follow as in the proof of Lemma \ref{closet} that $D(S)$ is a continue function, since it is a polynomial function.

\hfill \qed

\begin{lema} The set of the facial sextuples such that $D(S) \geq 0$ is closed.
    
\end{lema}\label{DScloset}
\noindent \textbf{Proof:} 

Recall that the intersection of two closet sets is a also closet. Then using same reasoning as in  Lemma \ref{closet} and Lemma \ref{DS2} the lemma is proved.
\hfill \qed

\begin{teor} \label{teorema1}
Let $\cal{T}$ be the set of all tetrahedral sextuples in $\R^6$. Consider the set $\cal{Z}\subset \cal{T}$ of all degenerated tetrahedral sextuples. Then $\cal{Z}$ is the boundary of $\cal{T}$.
\end{teor}

\noindent \textbf{Proof:}

From Lemma \ref{DScloset} we have that if $T^*$ is a non-tetrahedric sectuple, then $T^*$ does not belong to the interior of the boundary of $\cal{T}$. This is, $T^*$ belongs to the exterior of $\cal{T}$.

On the other hand, if $S$ is a nondegenerate tetrahedral sextuple, then $g_{ijk}(S) > 0$ for all $g_{ijk}$ related to $S$. Moreover, $D(S) >0$. We have also that $D(S)$ is continuous. Recall that in a real continuous function, if it is positive in a given point, it is also positive in a neighborhood of this point. Then $S$ cannot belong to the boundary of $\cal{T}$, but to the interior of $\cal{T}$.

Up to this point, we have that the boundary of $\cal{T}$ is formed solely by non-degenerate tetrahedral sextuples. To end up the proof we need to prove that any arbitrary degenerated tetrahedral sextuple is an accumulation point.

Let $T^*$ a degenerated tetrahedral sextuple. Then $D(T^*)=0$. By the continuity of $D$ in $\R^6$, $T^*$ does not belong to the interior of $\cal{T}$.
$T^*$ does not also belong to the complement of $\cal{T}$, then $T^*$ if an accumulation point of $\cal{T}$ and it belongs to the boundary of $\cal{T}$.
\hfill \qed

\section{Bisection of sextuples in $\R^6$}


From this point forward, we will restrict our analysis to the bisection of tetrahedral sextuples, which are the only ones that correspond to real tetrahedra. The bisection of non-tetrahedral sextuples can result in sextuples with negative elements, thus violating both the faciality condition and the requirement of a determinant greater than or equal to zero.

In order to represent the LEB of a sextuple $\textbf{x}$, let us consider the following linear mappings:
 \begin{equation}
 b_0 (\textbf{x}) = \frac{1}{4} \cdot \big( x_1,\,\,4x_2,\,\,2x_2+2x_3-x_1,\,\,2x_4+2x_5-x_1,\,\,4x_5,\,\,4x_6 \big)
 \label{formula_b0}
 \end{equation}
\begin{equation}
b_1(\textbf{x})=\frac{1}{4} \cdot \big( x_1,\,\,2x_2+2x_3-x_1,\,\,4x_3,\,\,4x_4,\,\,2x_4+2x_5-x_1,\,\,4x_6 \big)
 \label{formula_b1}
\end{equation}

Note that if $\textbf{x}$ is a tetrahedral sextuple, then $b_0(\textbf{x})$ and $b_1(\textbf{x})$ represents  
the two tetrahedral sextuples generated by bisecting the tetrahedron $\textbf{T}(\textbf{x})$ by the edge related to $x_1$, which implies that  
$b_0(\textbf{x})$ y $b_1(\textbf{x})$ are also tetrahedral sextuples. The primary benefit of using the squares of the edge lengths as the components of the sextuple is that the LEB equations become linear, significantly simplifying their computational calculation in terms of both efficiency and speed.

We denote by $\cal{P}$ the group of $6 \times 4 = 24$ valid permutations which transforms a tetrahedral sextuple into another tetrahedral sextuple, both related to the same tetrahedron. There are 24 possible permutations because there are 6 alternatives to choose the first edge, and from there, there are 4 alternatives to choose the second edge from among its 4 neighbors. The rest of the edges are determined because the third edge forms a face with the first two. The remaining three edges are the ones opposite to the first three in the order that was previously established. This implies, therefore, that a real tetrahedron has $24$ different sextuples in $\R^6$ that represent it.

We denote a permutation by a sequence of six numbers from $1$ to $6$ separated by a whitespace. For example, the permutation $(a_1\,\,a_2\,\,a_3\,\,a_4\,\,a_5\,\,a_6)$ represents that the element in the sextuple at position $a_1$ moves to the first position, the one at position $a_2$ to the second position, and so on.


Let $\Lambda$ be the set of multi-indices of finite length, whose elements belong to $\{0,\,1\}$. If $\alpha\in\Lambda$ and $i\in\{0,\,1\}$ we denote by $\overline{\alpha i}$ the resulting multi-index where the digit $i$ is added at the right of $\alpha$. 


Let us suppose that we take a function $f:\Lambda\rightarrow\cal{P}$. This way, in the case that $\textbf{x}$ is a tetrahedral sextuple, we call $P_\alpha=f(\alpha)$ to a valid permutation which puts in the first place the highest component of $\textbf{x}(\alpha)$.


Let \(\textbf{x}\) be an initial sextuple. If its first component is not the highest value, we choose a prior initial permutation to satisfy this requirement. Let us consider the following iterated process to be applied to \(\textbf{x}\): firsty, apply the LEB to obtain new sextuples $b_0(\textbf{x})$ and $b_1(\textbf{x})$. Secondly, apply a valid permutation to each sextuple to move its highest component to the first position. In this way, we obtain the sextuples $\textbf{x}(0)=P_0(b_0(\textbf{x}))$ and $\textbf{x}(1)=P_1(b_1(\textbf{x}))$. These two steps are repeated for every new sextuple doing $\textbf{x}(\overline{\alpha i})=(P_\alpha(b_i(\textbf{x}(\alpha)))$ for each $\alpha\in \Lambda$, $i\in\{0,\,1\}$, where $P_\alpha=f(\alpha)$.


This motivates a terminology for the function $f$. We will say that $f$ is a \emph{bisection pattern}, and then $f$ determines the component that will be bisected at each step. In the case that $\textbf{x}$ is a tetrahedral sextuple, and $P_\alpha$ is a valid permutation which puts in the first place the highest component of $\textbf{x}(\alpha)$, we have that the tetrahedron $\textbf{T}(\textbf{x}(\alpha))$ belongs to the orbit of the tetrahedron $\textbf{T}(\textbf{x})$ in the iterated LEB. Recall here that the orbit of a  tetrahedra is the set of tetrahedra generated during the iterative LEB.

For example, let $\textbf{x}=(15,13,12,11,10,14)$ be a tetrahedral sextuple, with its highest component already in the first position. To generate the first iteration of the LEB, we apply equations \ref{formula_b0} and \ref{formula_b1} to obtain sextuples $b_0(\textbf{x})=\frac{1}{4}(15,52,35,27,40,56)$ and $b_1(\textbf{x})=\frac{1}{4}(15,35,48,44,27,56)$ respectively. In order to move its highest component to the first position in both sextuples, we apply permutations $P_0=(6\,5\,2\,3\,4\,1)$ and $P_1=(6\,4\,3\,2\,5\,1)$, obtaining the following sextuples in this first iteration:
\begin{equation*}
\textbf{x}(0)=P_0(b_0(\textbf{x}))=\frac{1}{4}(56,40,52,35,27,15)
\end{equation*}
\begin{equation*}
\textbf{x}(1)=P_1(b_1(\textbf{x}))=\frac{1}{4}(56,44,48,35,27,15)
\end{equation*}

To generate the second iteration of the LEB, we can proceed in the same way, using the bisection equations with the sextuples \( \textbf{x}(0) \) and \( \textbf{x}(1) \), and obtaining valid permutations $P_{00}=(2\,3\,1\,5\,6\,4), P_{01}=(3\,2\,1\,4\,6\,5), P_{10}=(2\,3\,1\,5\,6\,4), P_{11}=(3\,2\,1\,4\,6\,5),$ for the four newly generated sextuples. Figure \ref{fig:FixedPattern} shows the first two levels of LEB applied to $\textbf{x}$. Continue this process and generate other levels in the tree using new valid permutations applied to the subdivided children.

\begin{figure}[!h]
\centering
\leavevmode 
\includegraphics[scale=.43]{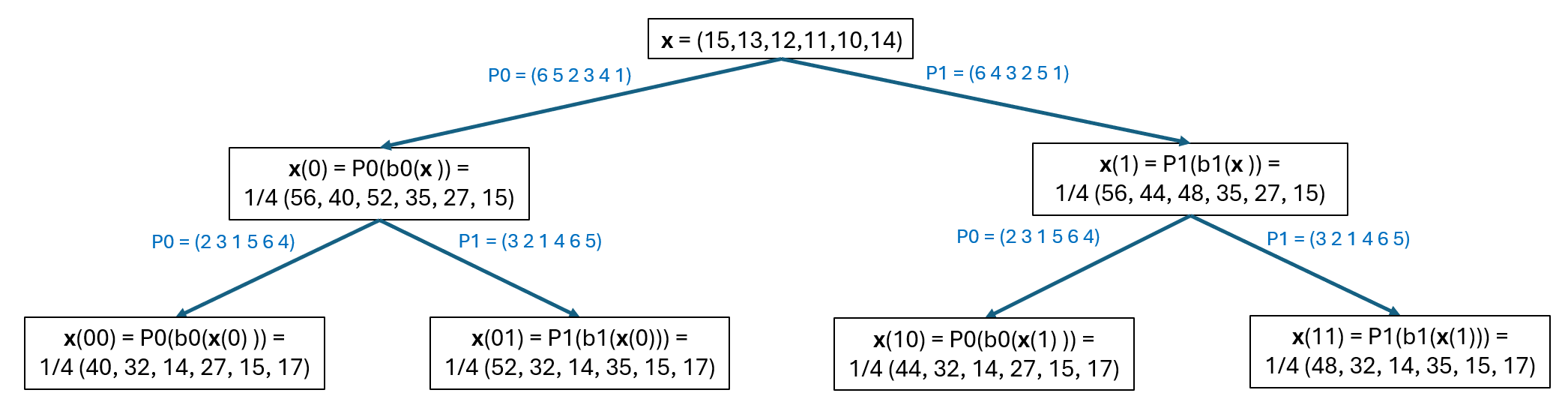}
\caption{Two levels of the LEB of $\textbf{x}=(15,13,12,11,10,14)$}
\label{fig:FixedPattern}
\end{figure}

To obtain the complete orbit we must take into account the case when the function $f$ is multivalued, i.e. $\textbf{x}(\alpha)$ has more than one maximun component. 



When the orbit of the tetrahedral sextuple $\textbf{x}$ in the iterated LEB generates only sextuples with an unique maximun component, we will say that $\textbf{x}$ is shown as \emph{uniform}. Otherwise we will say that  
$\textbf{x}$ shows a \emph{multiform} LEB pattern.
A multiform LEB pattern can be seen as a set of branches which are uniform LEB patterns, these are selections 
of the multivalued function $f$.

For example, let $\textbf{x} = (x_1, x_2, x_3, x_4, x_5, x_6) = (10, 16, 12, 16, 8, 6)$ be a tetrahedral sextuple produced during the LEB, in which the maximum value $16$ occurs twice, at $x_2$ and $x_4$. If we consider the permutation $P=(2\,3\,1\,5\,6\,4)$, that moves $x_2$ to the first position, we obtain $P(\textbf{x})=(16,12,10,8,6,16)$. And the two sextuples generated by the equations of LEB are:
\begin{equation*}
\textbf{x}(0)=b_0(P(\textbf{x}))=(4,12,7,3,6,16)
\end{equation*}
\begin{equation*}
\textbf{x}(1)=b_1(P(\textbf{x}))=(4,7,10,8,3,16)
\end{equation*}

This bisection corresponds to a first branch of the LEB pattern. But a second branch exists for $\textbf{x}$, if we use a different permutation $P=(4\,1\,5\,6\,3\,2)$, that moves $x_4$ to the first position. In this case, we obtain $P(\textbf{x})=(16,10,8,6,12,16)$, and the two new sextuples generated by the equations of LEB are: 
\begin{equation*}
\textbf{x}(0)=b_0(P(\textbf{x}))=(4,10,5,5,12,16)
\end{equation*}
\begin{equation*}
\textbf{x}(1)=b_1(P(\textbf{x}))=(4,5,8,6,5,16)
\end{equation*}

As can be seen, not only the sextuples but the tetrahedra generated in this second branch are completely different from those obtained in the first branch.



Another example of multiform LEB is the case of a regular tetrahedron, represented for example by the sextuple $\textbf{x}=(4,4,4,4,4,4)$. Note that initially, the choice of which component to subdivide is irrelevant, as the sextuples generated at the first level of the LEB are $(1, 4, 3, 3, 4, 4)$ and $(1, 3, 4, 4, 3, 4)$. However, at the second level of the LEB, three distinct branches emerge depending on which of the three longest edges is selected for the subdivision of each of the two sextuples mentioned above.
This is not the case of a 2D regular triangle, because it is indifferent which longest edge is chosen, as they result in equivalent results.


\section{Sharing a common LEB pattern}

We will say that two sextuples follow a \emph{LEB pattern} if all the sextuples generated in each of their respective orbits, corresponding to the same multi-index, have their maximum element in the same position. Moreover, in cases where they have more than one maximum element, these maxima must also appear in exactly the same positions.

Let $\mathbf{x}$ and $\mathbf{y}$ be two sextuples. They are said to follow the same pattern $f$ in the LEB if, when applying this pattern to both $\textbf{x}$ and $\textbf{y}$, the resulting $\mathbf{x_{\alpha}}$ and $\mathbf{y_{\alpha}}$ have their maximum component in the same position.  This is equivalent to stating that the sequence of permutations applied to each new sextuple generated in the orbit of the two initial sextuples is exactly the same. In order for two sextuples to follow the same pattern in the LEB, it is not sufficient that, at each step, the longest edge of both ends up in the same position. 

If the longest edge remains in the same position but the same permutations are not applied to both sextuples, then it is no longer the same pattern. However, as long as the longest edge remains in the same position, there exists a permutation that one can apply to both sextuples, thus constructing a common pattern for both.

For example, the sextuple $\textbf{y}=(19,17,14,15,16,18)$ follows the same \emph{LEB pattern} as the sextuple $\textbf{x}=(15,13,12,11,10,14)$ used in the previous subsection. By comparing Figures \ref{fig:FixedPattern} and \ref{fig:FixedPattern2}, we can see that the permutations used to place the highest component in the first position in each new sextuple generated during the iterative LEB process match exactly.

\begin{figure}[!h]
\centering
\leavevmode 
\includegraphics[scale=.43]{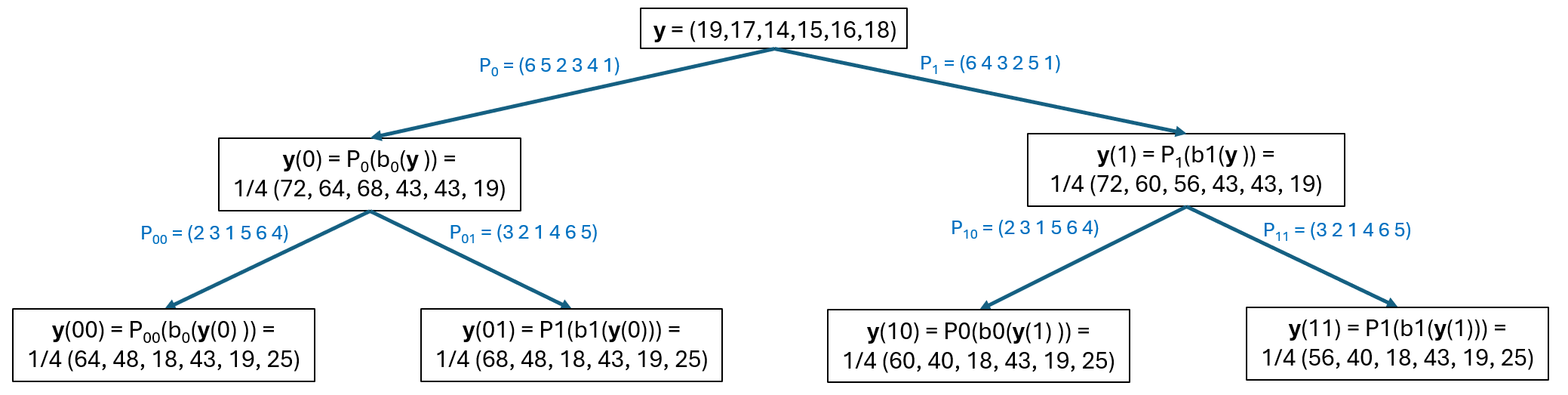}
\caption{Two levels of the LEB of $\textbf{y}=(19,17,14,15,16,18)$}
\label{fig:FixedPattern2}
\end{figure}

When referring to two tetrahedra instead of two sextuples, we will similarly say that both tetrahedra follow a \emph{LEB pattern} if there exist two sextuples, one associated with each tetrahedron, that follow the same pattern. This is analogous to stating that the sequence of permutations and subdivisions driven by the iterative LEB process in both tetrahedra is exactly the same.

We now intend to prove that the set of all sextuples following the same \emph{LEB pattern} is convex in $\mathbb{R}^6$. We recall here that a convex set in $\mathbb{R}^n$ is a set of points such that, for any pair of points within the set, the line segment connecting those two points is contained entirely within the set. 

As an example, let $\textbf{z}=(17,15,13,13,13,16)$ be a point in $\R^6$ lying on the line segment between the sextuples $\textbf{x}$ and $\textbf{y}$ mentioned earlier, which followed the same \emph{LEB pattern}. Specifically, 
it is the middle point between \textbf{x} and \textbf{y}.
Therefore, it can be shown that 
$\textbf{z}$ also follows the same \emph{LEB pattern}. In Figure \ref{fig:FixedPattern3} we can observe that the sequence of permutations in the LEB of $\textbf{z}$ is exactly the same as those used in the LEBs of $\textbf{x}$ and $\textbf{y}$.

\begin{figure}[!h]
\centering
\leavevmode 
\includegraphics[scale=.43]{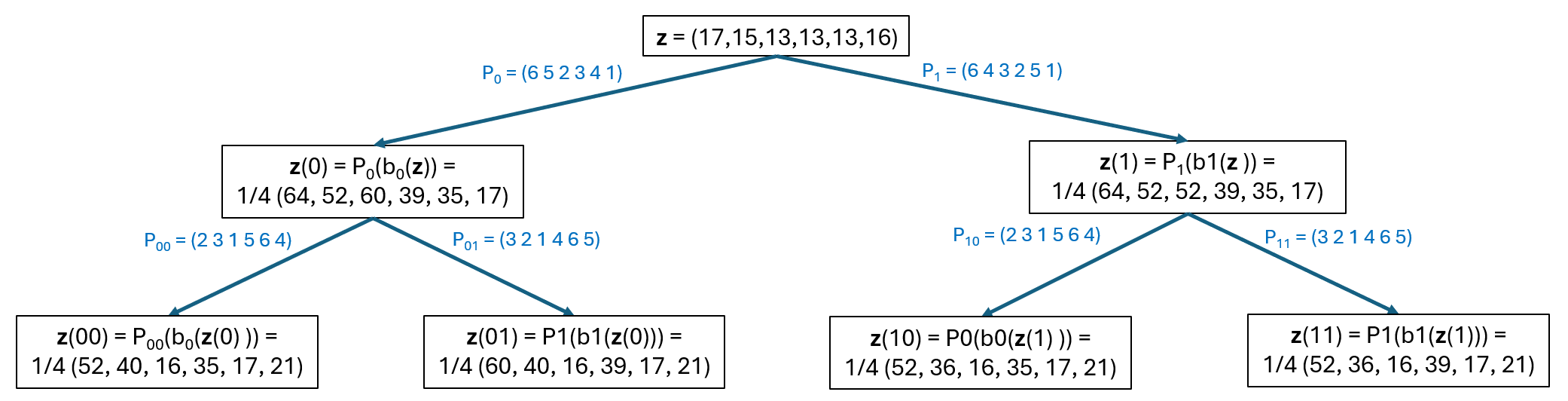}
\caption{Two levels of the LEB of $\textbf{z}=(17,15,13,13,13,16)$}
\label{fig:FixedPattern3}
\end{figure}


\begin{lema}
\label{lemmaConvex}
Let $f:\Lambda\rightarrow\cal{P}$ a given LEB pattern. 
The set $S$ of all sextuples that follow the pattern $f$ the LEB is a convex set in $\R^6$. 

\end{lema}

\noindent \textbf{Proof:}
Recall that a set $S\subset\R^6$ is convex if and only if for all pair of elements $\textbf{x}\subset\R^6$  and $\textbf{y}\subset\R^6$, we have that the line segment $L$ connecting $\textbf{x}$ and $\textbf{y}$ lies in $S$.
Furthermore we have $L=\{\textbf{z}\in\R^6:\,\textbf{z}=\delta\textbf{x}+(1-\delta)\textbf{y},\,\delta\in [0,\,1]\}$.

Suppose that $\textbf{x},\textbf{y}$ share a given pattern $f$ and let $\textbf{z}=\delta\textbf{x}+(1-\delta)\textbf{y},\,\delta\in [0,\,1]\}$.
For a given sextuple $\textbf{x}(\alpha)$ let $\textbf{x}(\alpha)_i$ be its  i-th component.
Let  $\max\limits_{i\in\{1,\cdots,6\}} (\textbf{x}(\alpha)_i)=m$ and $\max\limits_{i\in\{1,\cdots,6\}} (\textbf{y}(\alpha)_i)=l$ with $i\in\{1,\,2,\,\cdots,\,6\}$.\\ Then,

$$
\max\limits_{i\in\{1,\cdots,6\}} (\textbf{z}(\alpha)_i)=\max\limits_{i\in\{1,\cdots,6\}} (\{\delta\textbf{x}(\alpha)_i+(1-\delta)\textbf{y}(\alpha)_i\}) =\delta m+(1-\delta)l. 
$$
This completes the proof.
\hfill $\square$


\begin{lema} \label{lemmaNested}
Let $T$ be a degenerated tetrahedron with vertices $A,B,C,D$ such that $D$ is a point belonging to face $\bigtriangleup ABC$. Then the LEB of $T$ generates a sequence of nested tetrahedra such that $\{T\supset T_1 \supset T_2 \supset ... \supset T_k \}$ such that $T_k$ is generated at the step $k$, and $d(T_{3n}) \leq c \frac{d(T)}{2^n}$ for $d(T)$ being the diameter of $T$ and $c$ a constant non dependent of $n$.
\end{lema}

\noindent \textbf{Proof:}

At the first step of LEB, vertex $D$ belongs to any of the two new tetrahedra generated at this step. In Figure \ref{fig:tet}, $D \in \bigtriangleup ABC$ 

\begin{figure}[h]
\centering
\leavevmode 
\includegraphics[scale=.9]{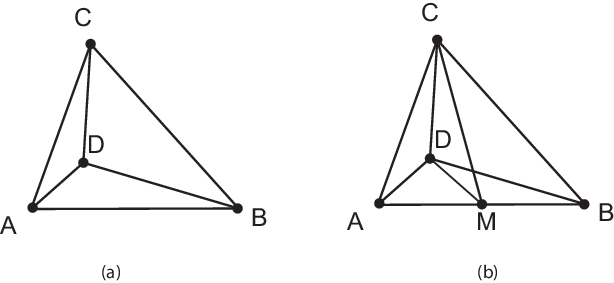}
\caption{(a) Degenerated tetrahedron $T$ with $D \in \bigtriangleup ABC$. (b) LEB of $T$. }
\label{fig:tet}
\end{figure}

Then $D$ belongs to any of two triangles generated by the LEB of $\bigtriangleup ABC$.

Let $T_1$ be a tetrahedron with vertices $A,B,M,C$. Repeat this process by taking $T_1$ instead of $T$ to obtain the sequence of new tetrahedra generated. It is clear that the diameter of $T_n$ is the same diameter of face of $T_n$ belonging to the set of triangles generated in the iterated LEB of $\bigtriangleup ABC$. Then it holds \[ d(T_{3n})  \leq c \frac{d(T)}{2^n} \]

\hfill $\square$

\begin{lema} Let $F$ be a non degenerated tetrahedron and $T$ a degenerated tetrahedron with vertices $A,B,C,D$ such that $D \in  \bigtriangleup ABC$. Then there not exists a common LEB pattern for $F$ and $T$.
\end{lema}
\noindent \textbf{Proof:}

Let us consider the sequence of nested tetrahedra like in Lemma \ref{lemmaNested} $\{T\supset T_1 \supset T_2 \supset ... \supset T_n \}$. Note that in the LEB of $T_n$ that obtain $T_{n+1}$, it is never bisected the edges connecting to vertex $D$. This is, to obtain $T_{n+1}$ it is not subdivided any edge containing $D$. In the LEB of a non degenerated tetrahedron this never happens.

Indeed, let $F=MNPQ$ where $M$, $N$, $P$ and $Q$ are the vertices of the non degenerated tetraedron $F$.
Let $F \supset F_1 \supset F_2 \supset ... \supset F_n$ be a sequence of nested tetrahedra generated in the iterated LEB of $F$, such that $F_n$ is generated at step $n$.

Let us suppose that every $F_n$ has a common vertex in $Q$ and lying in the opposite plane over the face $MNP$ of $F$. Then, all the edges containing the vertex $Q$ in the $F_n$ have lenghts greater or equal to the height $F$ in relation to face $MNP$. But, this contradicts the fact that, in the LEB of the tetrahedron, the maximum edge length obtained at each step tends to zero as $n$ approaches infinity.

Then, in a non degenerated tetrahedron, it cannot exist a sequence that hold $F \supset F_1 \supset F_2 \supset ... \supset F_n$ and this implies that the edges connecting $Q$ in the $F_n$ will be subdivided. 
Finally, this also implies that $F$ and $T$ cannot have a common pattern.

\hfill $\square$

\begin{lema} Let $T$ be a degenerated tetrahedron with quadrilateral shape. Then, the iterated LEB of $T$ generates at least, a new degenerated tetrahedron $Q$ of quadrilateral shape with vertices $E,F,G,H$, where $H \in  \bigtriangleup EFG$.
\end{lema}
\noindent \textbf{Proof:}

Let us see in Figure \ref{lema_1} a degenerated tetrahedron $T$ with quadrilateral shape given by vertices $A,B,C,D$ and $Q$ the degenerated tetrahedron given by vertices $E,F,G,H$.

\begin{figure}[h]
\centering
\leavevmode 
\includegraphics[scale=.9]{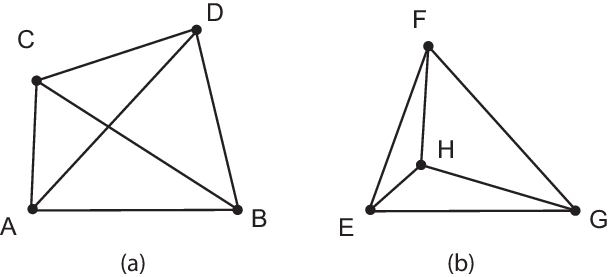}
\caption{(a) Tetrahedron $T$ with with quadrilateral shape. (b) Tetrahedron $Q$ of vertices $EFGH$ where vertex $H$ belongs to $\bigtriangleup EFG$}
\label{lema_1}
\end{figure}

Case 1: Suppose $T$ has as longest edge, one of the diagonal of the quadrilateral, Figure \ref{lema_2}. Then, the first step of the LEB generates two tetrahedra whose vertices are $B,C,D, M$ and $A,B,C,M$  see Figure \ref{lema_2} where $M$ is the midpoint of the longest edge $AD$ and $M$ belongs to one of the faces $ABC$, $BCD$, in this case, it belongs to $\bigtriangleup BCD$.

\begin{figure}[h]
\centering
\leavevmode 
\includegraphics[scale=.9]{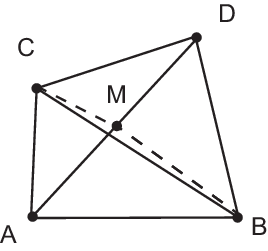}
\caption{Diagonal of quadrilateral tetrahedron $ABCD$ as its longest edge. }
\label{lema_2}
\end{figure}

 Case 2: Let $M$ the midpoint of longest edge $AB$. At the first step of LEB, it is generated two tetrahedra of quadrilateral shape $ACDM$ and $MCDB$.  It can be seen that the diagonals of $ABCD$ are just the diagonals in the quadrilaterals $ACDM$ and $MCDB$, see dashed lines in Figure \ref{lema1_3}.

\begin{figure}[h]
\centering
\leavevmode 
\includegraphics[scale=.9]{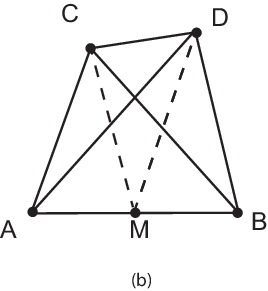}
\caption{LEB of tetrahedron $ABCD$ by the mid point $M$ }
\label{lema1_3}
\end{figure}

In the iterated LEB the diameters tend to zero, then this also hold in tetrahedra $ACDM$ and $MCDB$ instead of $ABCD$. Then, after a finite number of steps, we get a quadrilateral tetrahedron where its longest edge is just a diagonal of the generated tetrahedron. This leads us again to the initial case where the longest edge of the initial tetrahedron is one of its diagonals. This is the case where the triangular tetrahedron of the first step where found.

\hfill $\square$

\begin{teor} Let $\textbf{x}$ be a tetrahedral sextuple and $\textbf{y}$ a non-tetrahedral sextuple. Then there not exists a common LEB pattern for $\textbf{x}$ and $\textbf{y}$.
\end{teor}

\noindent \textbf{Proof:}
First, let us suppose that $\textbf{x}$ and $\textbf{y}$ share a common LEB pattern. Let us consider the segment of the line joining $\textbf{x}$ and $\textbf{y}$ in the space of sextuples. By the continuity of the \emph{generalized volume function}, we can find a point $\textbf{z}$ in that segment representing a degenerated sextuple. On the other hand, all the sextuples represented by points within this segment follow a common pattern in the iterated LEB. However, this cannot happen, by Lemma \ref{lemmaNested}. 
\hfill $\square$

\section{Directed graph formulation of the Longest Edge Bisection}

Permutations of sextuples and LEB patterns provide a convenient and rigorous framework for describing the iterative Longest Edge Bisection (LEB) in purely algebraic terms. In particular, when the sequence of permutations applied during the refinement is restricted to a finite set, the entire infinite genealogical tree generated by the LEB can be encoded as a finite directed graph. Each node of this graph represents a state, i.e., a specific pair of permutations that must be applied to the two children produced at a given bisection step, while directed edges describe the transitions between states induced by successive bisections.

In this section we show that both the $R_1^+$ family (which extends the Adler family of nearly equilateral tetrahedra) and the Liu–Joe family admit such a finite-state description. Remarkably, both families are governed by the same directed graph of eight states; the only difference lies in the initial state from which the iteration starts.  

\subsection{Finite-state graph induced by LEB patterns}

Recall that, for a tetrahedral sextuple $x$, the LEB produces two children given by the linear mappings $b_0(x)$ and $b_1(x)$ defined in Equations (4)–(5). In order to continue the iteration, each resulting sextuple must be permuted so that its largest component is placed in the first position. When this process is repeated iteratively, the sequence of permutations applied at each level defines a LEB pattern.

When the set of permutations that appear in this process is finite, the refinement can be represented by a directed graph $G = (K, Q, E, k_0)$, where:

\begin{itemize}
    
\item[-] $K$ is a finite set of states. Each state corresponds to a specific pair of permutations $(P_0, P_1)$ applied, respectively, to the children generated by $b_0$ and $b_1$.

\item[-] $Q = {0,1}$ is the set of input symbols, indicating whether the left ($0$) or right ($1$) child is selected.

\item[-] $E \subset K \times Q \times K$ is the set of directed edges, encoding the transition between states. If a path begins at $n$ and is directed to $m$, and the digit $q\in\{0,\,1\}$ is assigned to this path, this means that the permutation represented by node $m$ must be applied to all tetrahedra generated in $n$ whose end digit is $q$.

\item[-] $k_0 \in K$ is the initial state, determined by the initial sextuple.

\end{itemize}

Traversing the graph along a path labeled by a sequence of symbols in $Q$ specifies exactly which permutations must be applied at each refinement step. In this way, the infinite LEB tree is reduced to a finite automaton.

\subsection{Graph of the $R_1^+$ family}

We now consider the family $R_1^+$ introduced in \cite{AMC_2024}, which extends the classical Adler family of nearly equilateral tetrahedra. Adler in 1983, \cite{Adl} pointed out that if a tetrahedron is nearly equilateral (edge lengths within $5\%$ of each other) and the first and second longest edges are opposite, then the iterative LEB  produces $\leq 37$  similarity classes. Recently, Trujillo \textit{et. al}, \cite{AMC_2024} proved the conjecture given by Adler and improved the bound of $5\%$ to $22.47\%$, and this yields the $R_1^+$ fammily.

For tetrahedra belonging to $R_1^+$, the sequence of permutations required to place the longest edge first at each subdivision step coincides with the patterns illustrated in Figures \ref{fig:FixedPattern},\ref{fig:FixedPattern2} and \ref{fig:FixedPattern3}. In particular, the same finite set of permutation pairs reappears throughout the refinement.

For example, let $\textbf{x} = (15, 13, 12, 11, 10, 14)$ be the same tetrahedral sextuple already used in Section~3. The first two levels of the LEB applied to this sextuple were shown in \ref{fig:FixedPattern}. This sextuple belongs to the $R_1^+$ family, and its largest component is already in the first position. We can therefore group each sextuple in the refinement tree according to the permutations applied to its two children, assigning each of them to a state of the graph.

Figure \ref{fig:arbol_grafo_Adler} extends the LEB tree of the previous figure by showing additional refinement levels for several sextuples, and colors each node according to the permutations that must be applied to its children and further descendants. Each color thus represents a different state in the graph, which encodes the pair of permutations to be applied to the two children of the corresponding sextuple.

\begin{figure}[!h]
\centering
\leavevmode 
\includegraphics[width=\textwidth]{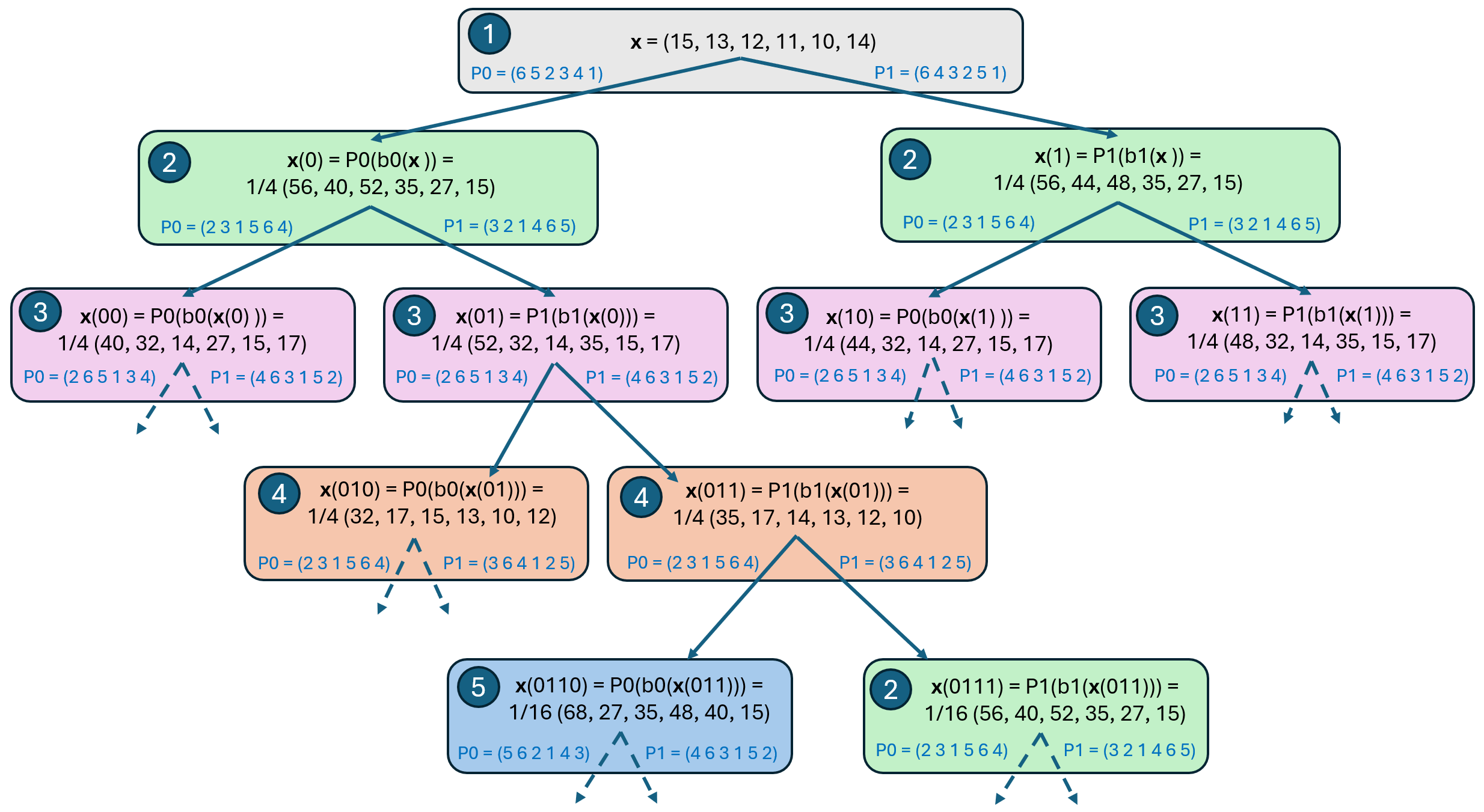}
\caption{LEB for the Adler tetrahedron $\textbf{x}=(15,13,12,11,10,14)$. States are circled by an integer and colors are used to differentiate the states.}
\label{fig:arbol_grafo_Adler}
\end{figure}

This observation leads to the following result.

\begin{lema}
The LEB of any tetrahedron in the $R_1^+$ family can be represented by a directed graph with eight states, as shown in Figure \ref{fig:grafo}. The initial state is state $\circled{1}$.
\end{lema}

\begin{figure}[ht]
\centering
\leavevmode 
\includegraphics[scale=.7]{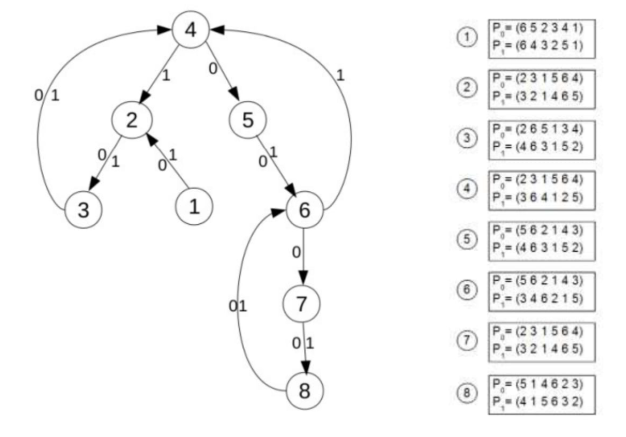}
\caption{LEB pattern of the family $R_1^+$}
\label{fig:grafo}
\end{figure}

The nodes of the graph represent the possible permutation pairs $(P_0, P_1)$ applied to the children produced by $b_0$ and $b_1$, while edges labeled by $0$ or $1$ indicate the corresponding transitions. In this way, all the infinitely many levels of the LEB tree of a representative $R_1^+$ tetrahedron (Figure \ref{fig:arbol_grafo_Adler}) collapse into this finite-state graph with eight states.

\begin{remark}
The directed graph representation provides a compact and efficient description of the refinement process. Once the graph is known, arbitrarily deep LEB refinements can be performed without recomputing permutation choices, which is particularly advantageous for large-scale or highly iterative computations.
\end{remark}

\subsection{Convexity of the $R_1^+$ family in $\mathbb{R}^6$}

An important structural property of LEB patterns, established in Lemma 4, is that the set of sextuples sharing the same LEB pattern forms a convex subset of $\mathbb{R}^6$.

Applying this result to the $R_1^+$ family yields the following interpretation. Consider sextuples representing tetrahedra in $R_1^+$ whose largest component appears in a fixed position of the sextuple. All such sextuples follow the same LEB pattern and therefore form a convex region in $\mathbb{R}^6$. Since the largest component may appear in any of the six positions, the full $R_1^+$ family can be expressed as the union of six convex regions, each corresponding to one possible position of the longest edge.

This convex decomposition provides a geometric description of the $R_1^+$ family directly in sextuple space, independently of any coordinate-based representation of tetrahedra.

\subsection{Same graph for the Liu–Joe family}

Besides the Adler and $R_1^+$ families, another important class of tetrahedra associated with the LEB is the Liu–Joe family \cite{LiuJoe-1994,LiuJoe-1995}. In the Liu–Joe method, a tetrahedron is first transformed affinely, then bisected using the LEB, and finally mapped back by the inverse transformation. For tetrahedra in the Liu–Joe family, this procedure yields exactly the same result as applying the LEB directly.

This lead us to identify a family of tetrahedra, called here the Liu-Joe family, consisting of the set of tetrahedra for which applying the Liu-Joe method yields the same result as directly applying the LEB method. In other words, a tetrahedron belongs to the Liu-Joe family if its bisection using the Liu-Joe method coincides with that of the LEB. This is equivalent to state the family as the set of tetrahedra for which every step in the bisection scheme proposed by Liu and Joe yields tetrahedra whose longest edge is listed first. When a tetrahedron belongs to the Liu--Joe family allows one to bypass the affine normalization step without altering the refinement outcome, leading to simpler and more efficient LEB algorithm.

Consider the tetrahedron $\textbf{x} = (80, 41, 40, 32, 31, 30)$, belonging to the Liu-Joe family. We apply the LEB to this tetrahedron and proceed in the same way as in Figure \ref{fig:arbol_grafo_Adler}, grouping the sextuples in the refinement tree according to the permutations applied to their children and further descendants. The resulting structure is shown in Figure \ref{fig:arbol_grafo_LJ}.

 \begin{figure}[!h]
\centering
\leavevmode 
\includegraphics[width=\textwidth]{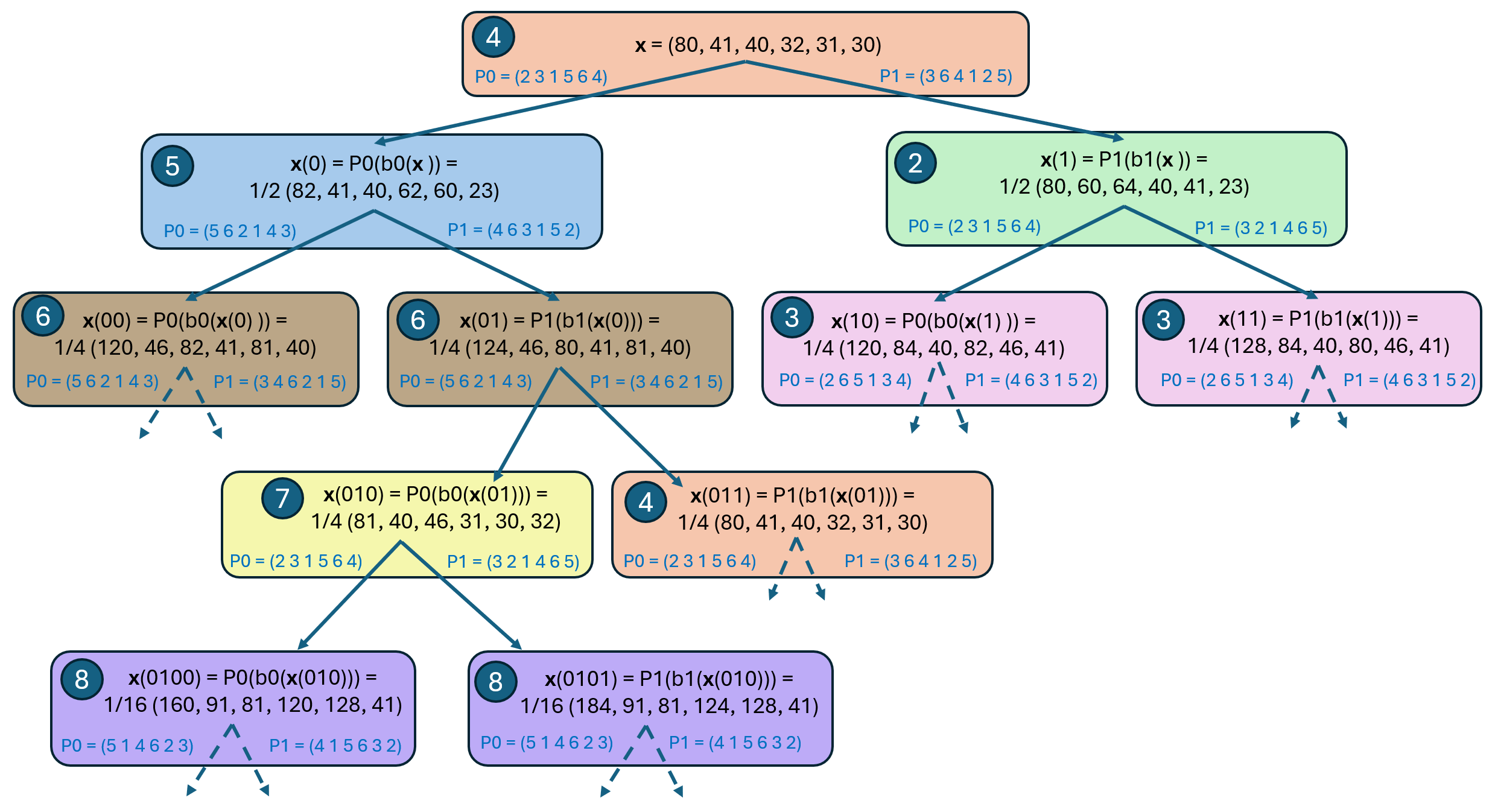}
\caption{
LEB for the Liu-Joe tetrahedron $\textbf{x}=(80,41,40,32,31,30)$. States are circled by an integer and colors are used to differentiate the states.}
\label{fig:arbol_grafo_LJ}
\end{figure}

As can be seen in this figure, the two children produced in the first LEB iteration apply the same pair of permutations corresponding to state \circled{4}. All subsequent descendants follow the same path in the graph starting from this state. Therefore, the graph that represents the LEB of tetrahedra in the $R_1^+$ family is exactly the same as the graph followed by the Liu-Joe family, but starting from state \circled{4} instead of state \circled{1}.

\subsection{Unified interpretation}

The previous results can be summarized as follows.

\begin{teor}
The directed graph with eight states shown in Figure \ref{fig:grafo} represents the LEB of both the $R_1^+$ family and the Liu–Joe family. For tetrahedra in $R_1^+$, the refinement starts at state \circled{1}, while for tetrahedra in the Liu–Joe family, it starts at state \circled{4}.
\end{teor}

This theorem shows that both families share the same underlying combinatorial structure for the LEB. The distinction between them is purely determined by the initial configuration, not by the graph itself.

As in the case of $R_1^+$, the Liu–Joe family can also be interpreted geometrically as a union of convex regions in $\mathbb{R}^6$, each region corresponding to a fixed position of the longest edge and therefore to a fixed LEB pattern. Although the two families are governed by the same directed graph, the corresponding convex regions in sextuple space are, in general, distinct.

Overall, the directed graph formulation provides a unified and efficient framework for understanding and implementing the LEB for important families of tetrahedra, highlighting the advantages of the sextuple representation for both theoretical analysis and practical computation.

\section{Conclusions}

We have developed a new algebraic and combinatorial formulation of the Longest Edge Bisection (LEB) of tetrahedra based entirely on sextuple representations in $\mathbb{R}^6$. By encoding tetrahedra through squared edge lengths, the refinement equations become linear, eliminating the need for coordinate-based computations and simplifying both theoretical and algorithmic treatment.

A key contribution of this work is the formalization of the multiform nature of the LEB in three dimensions. When equal longest edges occur, the refinement rule becomes intrinsically multivalued. We addressed this phenomenon by introducing the concept of Multiform LEB (MLEB) and by defining bisection patterns as permutation-driven descriptions of the refinement process.

We proved that sextuples sharing a common LEB pattern form convex subsets of $\mathbb{R}^6$, providing a geometric interpretation of refinement behavior directly in sextuple space. Furthermore, for the $R_1^+$ family and the Liu–Joe family, we showed that the infinite refinement tree collapses into a finite directed graph with eight states. Both families share exactly the same graph structure, differing only in their initial state, revealing a unified combinatorial backbone underlying their refinement dynamics.

The directed-graph formulation offers a compact and computationally efficient framework for analyzing deep LEB refinements and similarity-class evolution. Beyond its theoretical relevance, this approach opens new perspectives for discovering additional tetrahedral families with finite refinement behavior and for studying the geometry of refinement regions in $\mathbb{R}^6$.

Future work will focus on explicitly characterizing the regions of sextuple space corresponding to these families and on extending the multiform framework to more general refinement strategies.


\begin{thebibliography}{00}




\bibitem{Adl} A. Adler. On the bisection method for triangles.
{\it Math. Comp.} \textbf{40} (1983), 571--574. 

 

\bibitem{HaKoKri_2014}
A. Hannukainen, S. Korotov, M. K\v r\'{\i}\v zek. On numerical regularity of the face-to-face longest-edge bisection algorithm for tetrahedral partitions. Science of Computer Programming \textbf{90} (2014), 34--41.


 

\bibitem{Kos} 
I. Kossaczk{\'{y}}.
A recursive approach to local mesh refinement in two and three dimensions.
{\it  J. Comp. App. Math.} \textbf{55} (1994), 275--288.



\bibitem{LiuJoe-1994} 
A. Liu, B. Joe. 
On the shape of tetrahedra from bisection. 
{\it Math. Comp.} \textbf{63} (1994), 141--154. 




\bibitem{LiuJoe-1995} 
A. Liu, B. Joe. 
Quality of local refinement of tetrahedral meshes based on bisection. 
{\it SIAM J. Sci. Comput.} \textbf{16} (1995), 1269--1291.
 


\bibitem{BeGiRo} 
G. Belda-Ferr\'in, E. Ruiz-Giron\'es, X. Roca, Estimating the Number of Similarity Classes for Marked Bisection in General Dimensions. In International Meshing Roundtable,  2023 (271--289). Cham: Springer Nature Switzerland.



\bibitem{RivLev} M.-C. Rivara, C. Levin. A 3D refinement algorithm 
suitable for adaptive and  multigrid techniques.
{\it Comm. Appl. Numer. Methods Engrg.} \textbf{8} 
(1992), 281--290. 




\bibitem{Sik} K. Sikorski.
A three dimensional analogue to the method of
bisections for solving nonlinear equations.
{\it Math. Comp.} \textbf{33} (1979), 722--738. 


\bibitem{SuTrTa} 
J.\,P. Su\'arez, A. Trujillo, T. Moreno. 
Computing the Exact Number of Similarity Classes in the Longest Edge Bisection of Tetrahedra
{\it Mathematics} \textbf{9}(12) (2021), 1447. 

\bibitem{AMC_2024}
A. Trujillo-Pino, J.\,P. Su\'arez,  M. A Padr\'on. Finite number of similarity classes in Longest Edge Bisection of nearly equilateral tetrahedra. {\it Applied Mathematics and Computation}, \textbf{472} (2024) 128631.

\bibitem{PadronTrujilloSuarez2025}
M.~Á.~Padrón Medina, A.~R.~Trujillo Pino, J.~P.~Suárez,
\newblock Convergence of the $R_1^{+}$ tetrahedra family in iterative Longest Edge Bisection,
\newblock {\em Mathematics and Computers in Simulation}, \textbf{238} (2025), pp.~555--5667.


\bibitem{Casado2015}
G. Aparicio,  L. Casado, E. Hendrix, B G.-T\'oth, I. Garcia
On the minimum number of simplex shapes in longest edge bisection refinement of a regular $n$-simplex.
{\it Informatica} \textbf{26}(1) (2015), 17--32.

\bibitem{WirDre}
K. Wirth,  A. Dreiding. 
Edge lengths determining tetrahedrons.
{\it Elem. Math.} \textbf{64}(4) (2009), 160--170.

\end{thebibliography}
\end{document}